\documentclass[journal]{IEEEtran}

\usepackage{inputenc}
\usepackage{mathtools}
\usepackage{color}
\usepackage{array}
\usepackage{verbatim}
\usepackage{float}
\usepackage{amsmath}
\usepackage{amsthm}
\usepackage{amssymb}
\usepackage{graphicx}
\usepackage{longtable}
\usepackage{multirow}
\usepackage{balance}
\usepackage{mathtools}

\usepackage[unicode=true,
bookmarks=false,
breaklinks=false,pdfborder={0 0 1},colorlinks=false]
{hyperref}
\hypersetup{
	colorlinks,bookmarksopen,bookmarksnumbered,citecolor=red,urlcolor=red}
\usepackage{cite}

\usepackage{indentfirst}
\usepackage{textcomp}
\usepackage[T1]{fontenc}
\usepackage{amsthm}
\usepackage{amsmath,amstext,amsfonts}
\usepackage{amssymb}
\usepackage{mathtools}
\usepackage[cmintegrals]{newtxmath}
\usepackage{bm}  
\usepackage{siunitx}
\usepackage{physics}

\usepackage{graphicx}
\graphicspath{%
	{Figures/}
	{Figures/TiKz/}
	{Figures/Ipe/}
	{Figures/XFig/}
	{Figures/Inkscape/}
	{Figures/Matlab/}
	{Figures/Scanned/}
	{Figures/Saved/}
}
\usepackage{tikz}
\usepackage[list=true,labelformat=simple]{subcaption}
\usepackage{cite}
\usepackage{url}
\usepackage{xcolor}
\usepackage{siunitx}

\usepackage[inline]{enumitem}

\usepackage{algorithm,algpseudocode,float}
\usepackage{setspace} 

\algnewcommand{\algorithmicand}{\textbf{ and }}
\algnewcommand{\algorithmicor}{\textbf{ or }}
\algnewcommand{\AlgAnd}{\algorithmicand}
\algnewcommand{\AlgOr}{\algorithmicor}

\usepackage[noabbrev]{cleveref}  
\Crefname{figure}{Fig.}{Figs.}

\usepackage{tabularx,booktabs}
\newcolumntype{C}{>{\centering\arraybackslash}X} 
\usepackage{lipsum}

\makeatletter
\let\oldforeign@language\foreign@language
\DeclareRobustCommand{\foreign@language}[1]{%
	\lowercase{\oldforeign@language{#1}}}

\floatstyle{ruled}
\newfloat{algorithm}{tbp}{loa}
\providecommand{\algorithmname}{Algorithm}
\floatname{algorithm}{\protect\algorithmname}

\let\oldforeign@language\foreign@language
\DeclareRobustCommand{\foreign@language}[1]{%
	\lowercase{\oldforeign@language{#1}}}

\ifCLASSINFOpdf
\else
\fi

\newtheorem{lem}{Lemma}

\newtheorem{thm}{Theorem}
\newtheorem{rem}{Remark}

\newtheorem{assum}{Assumption}
\ifCLASSINFOpdf
\else
\fi

\def\ps@IEEEtitlepagestyle{%
	\def\@oddhead{\parbox[t][\height][t]{\textwidth}{\centering \scriptsize
			Personal use of this material is permitted. Permission from the author(s) and/or copyright holder(s), must be obtained for all other uses. Please contact us and provide details if you believe this document breaches copyrights.\\
			\noindent\makebox[\linewidth]{}
		}\hfil\hbox{}}%
	\def\@evenhead{\scriptsize\thepage \hfil \leftmark\mbox{}}%
	\def\@oddfoot{\parbox[t][\height][l]{\textwidth}{
			\vspace{-20pt}{\rule{\textwidth}{0.4pt}}\\ \footnotesize			{\bf{\footnotesize\textcolor{red}{A. M. Ali, H. A. Hashim, and C. Shen, "MPC Based Linear Equivalence with Control Barrier Functions for VTOL-UAVs," The 2024 IEEE American Control Conference (ACC), Toronto, Canada, 2024.}}}\\
			\noindent\makebox[\linewidth]
		}\hfil\hbox{}}%
	\def\@evenfoot{\MYfooter}}

\makeatother
\begin{document}
	\bstctlcite{IEEEexample:BSTcontrol}

\title{MPC Based Linear Equivalence with Control Barrier Functions for VTOL-UAVs}

\author{Ali Mohamed Ali, Hashim A. Hashim, and Chao Shen
	\thanks{This work was supported in part by the National Sciences and Engineering Research Council of Canada (NSERC), under the grants RGPIN-2022-04937, RGPIN-2022-04940, DGECR-2022-00103 and DGECR-2022-00106.}
	\thanks{A. M. Ali and H. A. Hashim are with the Department of Mechanical
		and Aerospace Engineering, Carleton University, Ottawa, ON, K1S-5B6,
		Canada (e-mail: AliMohamedAli@cmail.carleton.ca and hhashim@carleton.ca). C. Shen is with the Department of Systems and Computer Engineering, Carleton
		University, Ottawa, ON, K1S-5B6, Canada (shenchao@sce.carleton.ca).}
}



\maketitle
\begin{abstract}
	In this work, we propose a cascaded scheme of linear Model prediction Control (MPC) based on Control Barrier Functions (CBF) with Dynamic Feedback Linearization (DFL) for Vertical Take-off and Landing (VTOL) Unmanned Aerial Vehicles (UAVs). CBF is a tool that allows enforcement of forward invariance of a set using Lyapunov-like functions to ensure safety. The First control synthesis that employed CBF was based on Quadratic Program (QP) that modifies the existing controller to satisfy the safety requirements. However, the CBF-QP-based controllers leading to longer detours and undesirable transient performance. Recent contributions utilize the framework of MPC benefiting from the prediction capabilities and constraints imposed on the state and control inputs. Due to the intrinsic nonlinearities of the dynamics of robotics systems, all the existing MPC-CBF solutions rely on nonlinear MPC formulations or operate on less accurate linear models. In contrast, our novel solution unlocks the benefits of linear MPC-CBF while considering the full underactuated dynamics without any linear approximations. The cascaded scheme converts the problem of safe VTOL-UAV navigation to a Quadratic Constraint Quadratic Programming (QCQP) problem solved efficiently by off-the-shelf solvers. The closed-loop stability and recursive feasibility is proved along with numerical simulations showing the effective and robust solutions.
\end{abstract}

\begin{IEEEkeywords}
	Unmanned Aerial Vehicles, Vertical Take-off and Landing, Model Predictive Control, MPC, Nonlinearity, Dynamic Feedback Linearization, Optimal Control.
\end{IEEEkeywords}

\section{Introduction}\label{sec1}

The  urban applications of Vertical Take-off and Landing (VTOL) Unmanned
Aerial Vehicles (UAVs) highlight the significance of obstacle avoidance and hence classify VTOL-UAV navigation as a safety-critical system (visit Fig.\ref{fig: task}) \cite{hashim2023exponentially,hashim2023observer}.
The term safety-critical is used by researchers to distinguish systems
for which safety is a major design consideration \cite{ames2019control}.
The study of safety in dynamical systems started in \cite{nagumo1942lage}, where the necessary and sufficient conditions of the invariant set were studied
\cite{blanchini1999set}. Recently, the Control Barrier Function (CBF)
was introduced as a Lyapunov-like function that represents a safety
measure of the system \cite{ames2019control,ames2016control}. CBF-based safety constraints
were then utilized to construct a Quadratic Program (QP) minimizing
the difference between the feedback controller and the applied controller
under the CBF inequality constraint. The CBF-QP approach has been adopted to solve many control problems such as lane keeping  \cite{ames2019control},
safe navigation of a quadcopter \cite{singletary2021comparative},
haptic teleoperation \cite{zhang2020haptic}, and others. 
\begin{figure}[!htb]
	\centering\includegraphics[width=10cm,height=8cm,keepaspectratio]{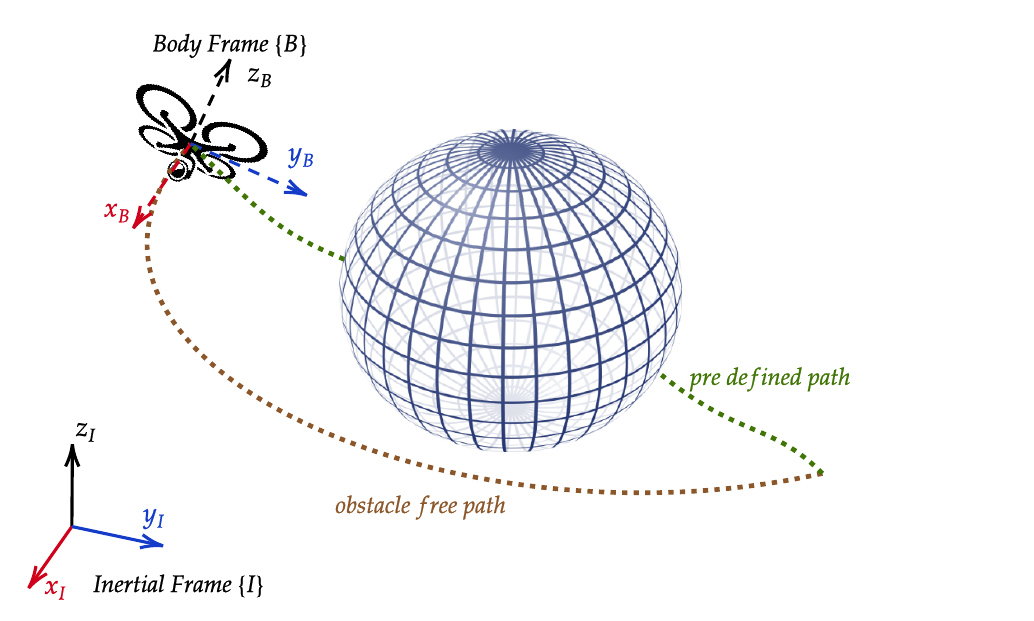}
	
	\caption{\label{fig: task} VTOL-UAV safe navigation task where $\{I\}$ is
		the body fixed frame and $\{B\}$ is the Global frame}
\end{figure}

The advantage of MPC, also known as receding horizon control, consists
in the strong ability to handle input and states constraints.
Thus, it is widely used for many robotic systems. In \cite{yoon2009model,frasch2013auto}
safety was considered in the MPC framework as an extra distance constraint
with various Euclidean norms in the optimization problem. However,
the obstacle avoidance constrains used by these approaches get activated
and change the robot's trajectory only if the robot is in close proximity
to an obstacle. As such, avoiding obstacles located further away requires
larger prediction horizon leading to higher computational time \cite{zeng2021safety}.
Nonetheless, the majority of existing literature uses a simplified
model of the dynamics, for instance \cite{turri2013linear} solves the lane keep problem using a linear passenger
car model to formulate the MPC optimization. Recently promising efforts have been made to combine MPC and CBFs. In \cite{son2019safety},
the authors proposed achieving vehicle safety by incorporating continuous-time
CBFs into discrete-time MPC. In \cite{rosolia2020multi}, the authors
proposed a control design that leverages CBFs for low-level control
and nonlinear MPC policies for high-level planning. This method treats
MPC and CBF separately utilizing them at different levels of the solution
to control the nonlinear model of Segway and for legged robots in
\cite{grandia2021multi}. In \cite{zeng2021safety},
a discrete formulation of the CBF was unified with nonlinear MPC and
applied to the car racing problem. In \cite{liu2023iterative}, the MPC based CBF was converted to an iterative convex optimization by linearizing   the  nonlinear dynamics as well as the CBF.
Collision avoidance for VTOL-UAVs has been investigated in \cite{thanh2018completion}
where a combination of geometrical constraints and kinematics was
used.
In \cite{chen2013uav} the authors proposed a similar framework
of the artificial potential fields, where Tangent Vector Field Guidance
(TVFG) algorithms employed to generate obstacle-free path for VTOL-UAVs.
Vision-based detection and collision avoidance have been widely studied
\cite{odelga2016obstacle}. 
The work in \cite{odelga2016obstacle}
proposed an MPC-based control with vision-based detection of obstacles.
However, the shortcoming of vision-based MPC approaches for collision
avoidance consists in VTOL-UAV starting to deviate from an obstacle
only when it is in its close proximity requiring larger
prediction horizon.

\paragraph*{Contributions} The proposed solution
bridges the above-identified literature gaps by formulating MPC using CBF to control the VTOL-UAV such that the safety constraint gets activated
everywhere providing smooth trajectory with the constraints are likely to be satisfied compared to usage of Euclidean norms as a constraint. The CBF provides
the notion of the global forward invariance of the safe set. The VTOL-UAV avoids the obstacle even if it is far from it leading to a shorter prediction  horizon, unlike the Euclidean distance case where the constraint is activated once it near the obstacle. Our proposed MPC-CBF with DFL
considers the full nonlinear underactuated UAV model. To address the system
nonlinearities and the computational complexity arising from the nonlinear
MPC, we introduce a solution that integrates a cascaded scheme of
Dynamic Feedback Linearization (DFL) with MPC unlocking the benefits
of linear MPC as opposed to nonlinear MPC in terms of computational burden and stability guarantees. The contribution of this
work can be summarized as follows:
\begin{enumerate}
	\item A cascaded scheme of DFL and
	MPC is proposed to unlock the benefits associated with Linear MPC.
	\item The mapping between the original nonlinear underactuated
	model and the linear equivalent model (MPC-CBF formulation
	designed on linear equivalent model by DFL) is presented.
	Combining MPC-CBF and DFL into a single scheme allows us to convert
	the safe navigation of the VTOL-UAV full model into a Quadratic Constraint
	Quadratic Programming (QCQP) Problem. 
	\item Closed Loop stability and recursive feasibility of the proposed scheme is proved and
	numerical simulations for a standard safe navigation task of a VTOL-UAV
	are carried out demonstrating the effectiveness and robustness of the proposed scheme.
\end{enumerate}
The remainder of the paper
is organized as follows: Section \ref{sec:Preliminaries} presents
preliminaries and mathematical notation. Section \ref{sec:model} introduces the VTOL-UAV dynamic model. The linear equivalence model is discussed in Section
\ref{sec:Sec2_Linear-Equivalence}. Section \ref{sec:Sec5_Proposed-Scheme}
presents the proposed control scheme. Section \ref{sec:Numerical-Results}
illustrates the effectiveness of the proposed scheme through numerical
simulations. Finally, Section \ref{sec:Sec7_Conclusion} concludes
the work.

\section{Preliminaries \label{sec:Preliminaries}}

In this paper, $\mathbb{R}$ denotes the set of real numbers 
and $\mathcal{L}$ defines the Lie derivative
operator. For a given vector field $f(x)$ such that $f:\mathbb{R}^{n}\rightarrow\mathbb{R}^{n}$
and a scalar function $\lambda:\mathbb{R}^{n}\rightarrow\mathbb{R}$,
the Lie derivative of $\lambda$ with respect to $f$ can be written
as $\mathcal{L}_{f}\lambda=\frac{\partial\lambda}{\partial x}\cdotp f(x).$
Consider the following single input single output nonlinear affine
in control system:
\begin{equation}
	\begin{cases}
		\dot{x} & =f(x)+g(x)u\\
		y & =h(x)
	\end{cases}\label{eq:pr1}
\end{equation}
where $x\in\mathbb{R}^{n}$ describes the system states, $u\in\mathbb{R}$
defines the system control input, $y\in\mathbb{R}$ denotes the system
output, $f:\mathbb{R}^{n}\rightarrow\mathbb{R}^{n}$, $g:\mathbb{R}^{n}\rightarrow\mathbb{R}^{n}$,
and $h:\mathbb{R}^{n}\rightarrow\mathbb{R}$. The relative degree
$r$ of such system can be defined at point $x_{0}$ if $\mathcal{L}_{g}\mathcal{L}_{f}^{\rho}h(x)=0$
for all $x$ in the neighborhood of $x_{0}$, $\rho<r-1$, and $\mathcal{L}_{g}\mathcal{L}_{f}^{r-1}h(x_{0})\neq0.$
The multi-input multi-output (MIMO) square affine in control system
is expressed as follows:

\begin{equation}
	\begin{cases}
		\dot{x} & =f(x)+g_{1}(x)u_{1}+\cdots+g_{m}(x)u_{m}\\
		y_{1} & =h_{1}(x)\\
		\vdots & =\vdots\\
		y_{m} & =h_{m}(x)
	\end{cases}\label{eq:pr2}
\end{equation}
where $x\in\mathbb{R}^{n}$ describes the system states, $u\in\mathbb{R}^{m}$
defines the system control input, $y\in\mathbb{R}^{m}$ denotes the system
output, $f:\mathbb{R}^{n}\rightarrow\mathbb{R}^{n}$, $g:\mathbb{R}^{n}\rightarrow\mathbb{R}^{m}$,
and $h:\mathbb{R}^{n}\rightarrow\mathbb{R}^{m}$.
\begin{lem}
	\label{lem:degree}\cite{isidori1985nonlinear} The 
	relative degree
	of \eqref{eq:pr2} at $x_{0}$ is described as $r=[r_{1},\ldots,r_{m}]^{\top}\in\mathbb{R}^{m}$
	exists if:
	\begin{itemize}
		\item $\mathcal{L}_{g_{j}}\mathcal{L}_{f}^{\rho}h_{i}(x)=0$ at the neighborhood
		of $x_{0}$ for all $1\leq j\leq m$, $\rho<r_{i}-1$, and $1\leq i\leq m$.
		\item The decoupling matrix $A(x)\in\mathbb{R}^{m \times m}$ defined as{\small
			\begin{equation}
				A(x)=\left(\begin{array}{ccc}
					\mathcal{L}_{g_{1}}\mathcal{L}_{f}^{r_{1}-1}h_{1}(x) & \cdots & \mathcal{L}_{g_{m}}\mathcal{L}_{f}^{r_{1}-1}h_{1}(x)\\
					\mathcal{L}_{g_{1}}\mathcal{L}_{f}^{r_{2}-1}h_{2}(x) & \cdots & \mathcal{L}_{g_{m}}\mathcal{L}_{f}^{r_{1}-1}h_{2}(x)\\
					\cdots & \cdots & \cdots\\
					\mathcal{L}_{g_{1}}\mathcal{L}_{f}^{r_{m}-1}h_{m}(x) & \cdots & \mathcal{L}_{g_{m}}\mathcal{L}_{f}^{r_{m}-1}h_{m}(x)
				\end{array}\right)\label{eq:decouplingmatrix-1}
			\end{equation}
		}is nonsingular at $x=x_{0}$.
	\end{itemize}
\end{lem}
\begin{lem}
	\label{lem:Full}\cite{isidori1985nonlinear} The input-to-state feedback
	linearization of the system dynamics in \eqref{eq:pr2} is solvable
	at $x_{0}$ using the control input $u=A^{-1}(x)\left(\left[\begin{array}{c}
		\mathcal{-L}_{f}^{r_{1}}h_{1}(x)\\
		\vdots\\
		-\mathcal{L}_{f}^{r_{m}}h_{m}(x)
	\end{array}\right]+\left[\begin{array}{c}
		v_{1}\\
		\vdots\\
		v_{m}
	\end{array}\right]\right)$, if $\sum_{r=1}^{m}r=n$ and the decoupling input matrix $A(x)$
	in \eqref{eq:decouplingmatrix-1} is of full rank. $v$ is an external
	reference input to be defined.
\end{lem}
Throughout this paper $\mathcal{H}(x)$ represents a control barrier function
with $\mathcal{H}(x):D\subset\mathbb{R}^{n}\rightarrow\mathbb{R}_{\geq0}$   
describing a safety metric. The class of $\mathcal{K_{\infty}}$ extended
function is denoted by $\kappa$ such that it is a strictly increasing function
with the mapping $[0,\infty)\rightarrow[0,\infty)$ and $lim_{r\rightarrow\infty}\kappa(r)=\infty$.
Define $C=\{x\in D\subset\mathbb{R}^{n}:\mathcal{H}(x)\geq0\}$ to
be the safe set. 

\section{VTOL-UAV Dynamic Model \label{sec:model}}

The Euler angles $\zeta=[\phi,\theta,\psi]^{\top}\in\mathbb{R}^{3}$ mapping
to rotat-ional matrix $R\in SO(3)$ follows $R_{\zeta}:\mathbb{R}^{3}\rightarrow SO(3)$
\cite{hashim2019special}:
\begin{equation}
	R_{\zeta}=\left[\begin{array}{ccc}
		c_{\psi}c_{\theta} & c_{\psi}s_{\phi}s_{\theta}-c_{\phi}s_{\psi} & s_{\phi}s_{\psi}+c_{\phi}c_{\psi}s_{\theta}\\
		c_{\theta}s_{\psi} & c_{\phi}c_{\psi}+s_{\phi}s_{\psi}s_{\theta} & c_{\phi}s_{\psi}s_{\theta}-c_{\psi}s_{\phi}\\
		-s_{\theta} & c_{\theta}s_{\phi} & c_{\phi}c_{\theta}
	\end{array}\right]\label{eq:1-2}
\end{equation}
with $c=\cos$, $s=\sin$, and $t=\tan$. The dynamics of a rigid-body
under external forces applied to the center of mass and expressed
in Newton-Euler formalism are defined as follows:
\begin{align}
	\sum F & =m\dot{V}\hspace{2.5cm}F,V\in\{\mathcal{I}\}\label{eq:18}\\
	\sum M & =J\dot{\Omega}+\Omega\times J\Omega\hspace{1cm}M,\Omega\in\{\mathcal{\mathcal{B}}\}\label{eq:19}
\end{align}
where $V=[v_{x},v_{y},v_{z}]^{\top}$ is the linear velocity vector,
$\Omega=[\Omega_{x},\Omega_{y},\Omega_{z}]^{\top}$ is rotational
velocity, $m$ is the total mass of the UAV, $J$ is the UAV inertia
matrix, $\{\mathcal{I}\}$ denotes inertial-frame, and $\{\mathcal{B}\}$
refers to body-frame. $\sum F$ and $\sum M$ are the
sum of external forces and moments acting on the UAV.
The full nonlinear underactuated model (neglecting the aerodynamic forces) can
be written in the form of MIMO affine in control system as in \eqref{eq:pr2}
with $x=[p_{x},p_{y},p_{z},\phi,\theta,\psi,v_{x},v_{y},v_{z},\dot{\phi},\dot{\theta},\dot{\psi}]^{\top}\in\mathbb{R}^{12}$
being a state vector and the functions $f(x)$, $g_{1}(x)$, $g_{2}(x)$,
$g_{3}(x)$, and $g_{4}(x)$ being expressed as follows \cite{bouabdallah2007design}:{\small
	\begin{alignat}{1}
		f(x) & =\left[\begin{array}{c}
			V\\
			\dot{\theta}s_{\phi}\sec_{\theta}+\dot{\psi}c_{\phi}\sec_{\theta}\\
			\dot{\theta}c_{\phi}-\dot{\psi}s_{\phi}\\
			\dot{\phi}+\dot{\theta}s_{\phi}t_{\theta}+\dot{\psi}c_{\phi}t_{\theta}\\
			0\\
			0\\
			-g\\
			\frac{I_{y}-I_{z}}{I_{x}}\dot{\theta}\dot{\psi}\\
			\frac{I_{z}-I_{x}}{I_{y}}\dot{\phi}\dot{\psi}\\
			\frac{I_{x}-I_{y}}{I_{z}}\dot{\phi}\dot{\theta}
		\end{array}\right],g_{1}(x)=\left[\begin{array}{c}
			0_{1\times6}\\
			\frac{-(c_{\phi}c_{\phi}s_{\theta}+s_{\phi}s_{\phi})}{m}\\
			\frac{-(c_{\phi}c_{\theta}s_{\psi}+s_{\psi}s_{\phi})}{m}\\
			\frac{c_{\theta}c_{\phi}}{m}\\
			0_{1\times3}
		\end{array}\right]\nonumber \\
		g_{2}(x) & =\left[\begin{array}{c}
			0_{1\times9}\\
			\frac{d}{I_{x}}\\
			0_{1\times2}
		\end{array}\right],g_{3}(x)=\left[\begin{array}{c}
			0_{1\times10}\\
			\frac{d}{I_{y}}\\
			0
		\end{array}\right],g_{4}(x)=\left[\begin{array}{c}
			0_{1\times11}\\
			\frac{d}{I_{z}}
		\end{array}\right]\label{eq:201}
	\end{alignat}
}where $V=[v_{x},v_{y},v_{z}]^{\top}$,  $d$ is the distance from the center of mass to the rotors,
and $g$ is the gravity constant. $u_{1}$ associated with $g_{1}(x)$
is the total applied thrust and $[u_{2},u_{3},u_{4}]$ associated
with $[g_{1}(x),g_{2}(x),g_{3}(x)]$ refer to UAV rotational torque
inputs. The output is expressed as follows:
\begin{equation}
	y_{1}=x,\hspace{0.25cm}y_{2}=y,\hspace{0.25cm}y_{3}=z,\hspace{0.25cm}y_{4}=\psi.\label{eq:25}
\end{equation}

\section{Linear Equivalence Model \label{sec:Sec2_Linear-Equivalence}}
In this section, the linear equivalence model of the VTOL-UAV
model in \eqref{eq:201} will be discussed. In view of \eqref{lem:degree}
and \eqref{lem:Full}, the VTOL-UAV nonlinear underactuated dynamics can be converted into a linear system of a chain
of integrators if the input-to-state full feedback Linearization is
solvable.
\subsection{Feedback Linearization}
Recalling \eqref{lem:degree}, the vector relative degree can be defined
as $r=[r_{1},r_{2},r_{3},r_{4}]^{\top}$. For $\rho=1,$ one shows that
$\mathcal{L}_{g_{j}}h_{i}(x)=0$, $\forall1\leq j\leq4$ and $1\leq i\leq4$,
and hence one concludes that $r_{1}=r_{2}=r_{3}=r_{4}\neq1$. For
$\rho=2$ and after some tedious calculations to compute $A(x)$ one
obtains:
{\small
	\begin{alignat}{1}
		A(x) & =\left[\begin{array}{cccc}
			\mathcal{L}_{g_{1}}\mathcal{L}_{f}h_{1}(x) & \mathcal{L}_{g_{2}}\mathcal{L}_{f}h_{1}(x) & \cdots & \mathcal{L}_{g_{4}}\mathcal{L}_{f}h_{1}(x)\\
			\mathcal{L}_{g_{1}}\mathcal{L}_{f}h_{2}(x) & \mathcal{L}_{g_{2}}\mathcal{L}_{f}h_{2}(x) & \cdots & \mathcal{L}_{g_{4}}\mathcal{L}_{f}h_{2}(x)\\
			\vdots & \vdots & \ddots & \vdots\\
			\mathcal{L}_{g_{1}}\mathcal{L}_{f}h_{4}(x) & \mathcal{L}_{g_{2}}\mathcal{L}_{f}h_{4}(x) & \cdots & \mathcal{L}_{g_{4}}\mathcal{L}_{f}h_{4}(x)
		\end{array}
		\right]\nonumber \\
		& =\left[\begin{array}{cccc}
			\frac{-(c_{\phi}c_{\phi}s_{\theta}+s_{\phi}s_{\phi})}{m} & 0 & 0 & 0\\
			\frac{-(c_{\phi}c_{\theta}s_{\psi}+s_{\psi}s_{\phi})}{m} & 0 & 0 & 0\\
			\frac{-c_{\theta}c_{\phi}}{m} & 0 & 0 & 0\\
			0 & 0 & \frac{d}{I_{y}}s_{\phi}sec_{\theta} & \frac{1}{I_{z}}c_{\phi}sec_{\theta}
		\end{array}\right]\label{eq:26}
\end{alignat}} 
It becomes apparent that $A(x)$ is singular and in view of \eqref{lem:Full}
the full input-to-state feedback linearization using the control input
is unsolvable. As such, some modifications are necessary to render
the nonlinear model dynamics in \eqref{eq:201} accounting for the
full input-to-state feedback linearizable form which is the focus
of the next subsection. 

\subsection{Dynamic Feedback Linearization (DFL)}

DFL also known as the dynamic extension algorithm, is comprehensively
discussed in (\cite{isidori1985nonlinear}, Chapter 5). By analyzing
the decoupling matrix $A(x)$ in \eqref{eq:26}, it is obvious that
the applied torque inputs $u_{2}$, $u_{3}$, and $u_{4}$ are the
main challenge since they do not appear in any of the first three
outputs ($x,y,z$) resulting in $A(x)$ singularity. One way to address
this is issue is by adding a differential delay to the total thrust,
$u_{1}$. This way, the integrator would allow $u_{2}$, $u_{3}$,
and $u_{4}$ to appear in the zero-column of $A(x)$ in \eqref{eq:26}.
By setting $u_{1}$ to be the output of two integrators adding two
new states such that $u_{1}=\zeta$, $\dot{\zeta}=\zeta$ and $\dot{\zeta}=\mathcal{U}_{1},$
resulting in a new vector fields $\bar{f}(\bar{x})$, $\bar{g}_{1}(\bar{x})$,
$\bar{g}_{2}(\bar{x})$, $\bar{g}_{3}(\bar{x})$ and $\bar{g}_{4}(\bar{x})$ as well
as new control signals $\mathcal{U}_{1}$ while the three control
torques will remain unchanged $\mathcal{U}_{2}=u_{2}$, $\mathcal{U}_{3}=u_{3}$,
and $\mathcal{U}_{4}=u_{4}$. The extended dynamics of the VTOL-UAV
will have a new state vector $\bar{x}=[p_{x},p_{y},p_{z},\phi,\theta,\psi,v_{x},v_{y},v_{z},\zeta,\dot{\zeta},\dot{\phi},\dot{\theta},\dot{\psi}]^{\top}\in\mathbb{R}^{14}$
and the new system have $\bar{n}=14$. 
		Based on \eqref{lem:Full}, $\sum_{r=1}^{4}\bar{r}=14=\bar{n}$ implies that the full input-to-state feedback linearization is solvable for
		the extended UAV dynamic model 
		using
		the feedback controller:{\small
			\begin{equation}
				u=\text{\ensuremath{\underbrace{\left[\begin{array}{cccc}
								\bar{a}_{11} & \bar{a}_{12} & \bar{a}_{13} & \bar{a}_{14}\\
								\bar{a}_{21} & \bar{a}_{22} & \bar{a}_{23} & \bar{a}_{24}\\
								\bar{a}_{31} & \bar{a}_{32} & \bar{a}_{33} & \bar{a}_{34}\\
								\bar{a}_{41} & \bar{a}_{42} & \bar{a}_{43} & \bar{a}_{44}
							\end{array}\right]}_{\bar{A}^{-1}(x)}}}\left(\left[\begin{array}{c}
					-\mathcal{L}_{\bar{f}}^{4}h_{1}(\bar{x})\\
					-\mathcal{L}_{\bar{f}}^{4}h_{2}(\bar{x})\\
					-\mathcal{L}_{\bar{f}}^{4}h_{3}(\bar{x})\\
					\mathcal{-L}_{f}^{2}h_{4}(\bar{x})
				\end{array}\right]+\left[\begin{array}{c}
					v_{1}\\
					v_{2}\\
					v_{3}\\
					v_{4}
				\end{array}\right]\right)\label{eq:25-1}
			\end{equation}
		}where:{\small
			\begin{flalign*}
				\small
				\bar{a}_{11} & =m(s_{\bar{x}_{4}}s_{\bar{x}_{6}}+c_{\bar{x}_{4}}c_{\bar{x}_{6}}s_{\bar{x}_{5}}) & \bar{a}_{31} & =\frac{I_{y}mc_{\bar{x}_{5}}c_{\bar{x}_{6}}}{dx_{13}}\\
				\bar{a}_{12} & =m(c_{\bar{x}_{4}}s_{\bar{x}_{5}}s_{\bar{x}_{6}}-mc_{\bar{x}_{6}}s_{\bar{x}_{4}}) & \bar{a}_{32} & =\frac{I_{y}mc_{\bar{x}_{5}}s_{\bar{x}_{6}}}{dx_{13}}\\
				\bar{a}_{21} & =\frac{I_{x}m(c_{\bar{x}_{4}}s_{\bar{x}_{6}}-c_{\bar{x}_{6}}s_{\bar{x}_{4}}s_{\bar{x}_{5}})}{d\bar{x}_{13}} & \bar{a}_{33} & =\frac{-I_{y}m\,s_{\bar{x}_{5}}}{d\bar{x}_{13}}\\
				\bar{a}_{22} & =\frac{I_{x}m(c_{\bar{x}_{4}}c_{\bar{x}_{6}}-s_{\bar{x}_{4}}s_{\bar{x}_{5}}s_{\bar{x}_{6}})}{d\bar{x}_{13}} & \bar{a}_{41} & =\frac{-I_{z}m\,c_{\bar{x}_{5}}c_{\bar{x}_{6}}t_{\bar{x}_{4}}}{d\bar{x}_{13}}\\
				\bar{a}_{23} & =\frac{I_{x}m\,c_{\bar{x}_{5}}s_{\bar{x}_{4}}}{d\bar{x}_{13}} & \bar{a}_{42} & =\frac{-I_{z}m\,c_{\bar{x}_{5}}s_{\bar{x}_{4}}s_{\bar{x}_{6}}}{d\bar{x}_{13}c_{\bar{x}_{4}}}\\
				\bar{a}_{43} & =\frac{-I_{z}m\,s_{\bar{x}_{5}}t_{\bar{x}_{4}}}{d\bar{x}_{13}} & \bar{a}_{44} & =\frac{-I_{z}mc_{\bar{x}_{5}}}{d\bar{x}_{4}}\\
				\bar{a}_{13} & =mx_{\bar{x}_{4}}c_{\bar{x}_{5}} & \bar{a}_{14} & =\bar{a}_{24}=\bar{a}_{34}=0
			\end{flalign*}
		}Thus, one can verify that the invertibility of the decoupling matrix is guaranteed by constraining
		the pitch angle $\theta$ such that $-\frac{\pi}{2}<\theta<\frac{\pi}{2}$ and the roll angle such that $-\frac{\pi}{2}<\phi<\frac{\pi}{2}$.
		The exact expressions of $\mathcal{L}_{\bar{f}}^{4}h_{1}(\bar{x})$,
		$\mathcal{L}_{\bar{f}}^{4}h_{2}(\bar{x})$, $\mathcal{L}_{\bar{f}}^{4}h_{3}(\bar{x})$
		and $\mathcal{L}_{\bar{f}}^{2}h_{4}(\bar{x})$ will not be stated
		due to page limitations.
		
		\section{Proposed Scheme \label{sec:Sec5_Proposed-Scheme}}
		
		The key feature of obstacle avoidance control scheme is the safety
		constraints represented by the CBF. In \cite{ames2019control}, the
		authors proposed CBF $\mathcal{H}_{k}$ representing a safety metric
		as the distance between the moving object and the obstacle. The sufficient
		and necessary conditions for the safe maneuvers are based on the usage
		of class $\mathcal{K_{\infty}}$ function similar to Lyapunov functions
		such that $\dot{\mathcal{H}}(x,u)\geq\mathcal{-\kappa}(\mathcal{H}(x))\Longleftrightarrow C$
		is invariant. The proposed scheme makes use of the CBF concepts in
		the MPC formulation cascaded by dynamic feedback controller defined
		in \eqref{eq:25-1}. The motivation for using the
		cascaded scheme will be unlocking the usage of a linear MPC with all
		its merit compared to the nonlinear MPC in terms of computational
		cost and ease of stability guarantees. 
		\begin{lem}
			\label{thm:thm1}(States  mapping of VTOL-UAV in MPC-DFL scheme) Recall the VTOL-UAV  model dynamics in \eqref{eq:201}.
			Using the control input in \eqref{eq:25-1}, the optimization problem
			of the MPC in the cascaded scheme of MPC-DFL can be formulated on
			a linear equivalent model dynamics as $\dot{z}=A_{z}z+B_{z}v$ with
			the following state and input  mapping:
			\begin{eqnarray}
				z_{1}= & x, & z_{8}=\dddot{y}\nonumber \\
				z_{2}= & \dot{x}, & z_{9}=z\nonumber \\
				z_{3}= & \ddot{x}, & z_{10}=\dot{z}\label{eq:26-1}\\
				z_{4}= & \dddot{x}, & z_{11}=\ddot{z}\nonumber \\
				z_{5}= & y, & z_{12}=\dddot{z}\nonumber \\
				z_{6}= & \dot{y}, & z_{13}=\psi\nonumber \\
				z_{7}= & \ddot{y}, & z_{14}=\dot{\psi}\nonumber 
			\end{eqnarray}
			\begin{equation}
			\left[\begin{array}{c}
				v_{1}\\
				v_{2}\\
				v_{3}\\
				v_{4}
				\end{array}\right]=\bar{A}(\bar{x})\left[\begin{array}{c}
				u_{1}\\
				u_{2}\\
				u_{3}\\
				u_{4}
				\end{array}\right]+\left[\begin{array}{c}
				\mathcal{L}_{\bar{f}}^{4}h_{1}(\bar{x})\\
				\mathcal{L}_{\bar{f}}^{4}h_{2}(\bar{x})\\
				\mathcal{L}_{\bar{f}}^{4}h_{3}(\bar{x})\\
				\mathcal{L}_{\bar{f}}^{2}h_{4}(\bar{x})
				\end{array}\right]\label{eq:27}
				\end{equation}
			\end{lem}
			\begin{proof}
				The detailed proof is omitted due to the limited page numbers. However, using the state
				mapping from $z$ to $x$ in \eqref{eq:26-1}, having the derivatives
				of $z$ states, and applying the controller in \eqref{eq:25-1}, one
				finds the following linear system:
				\begin{alignat}{1}
					\dot{z}= & A_{z}z+B_{z}v.\label{eq:28}\\
					y_{z}= & C_{z}z.\label{eq:29}
				\end{alignat}
				where
				\begin{align*}
					A_{z} & =\left[\begin{array}{cccc}
						A_{z1} & 0_{4\times4} & 0_{4\times4} & 0_{4\times2}\\
						0_{4\times4} & A_{z1} & 0_{4\times4} & 0_{4\times2}\\
						0_{4\times4} & 0_{4\times4} & A_{z1} & 0_{4\times2}\\
						0_{2\times4} & 0_{2\times4} & 0_{2\times4} & A_{z2}
					\end{array}\right],B_{z}=\left[\begin{array}{c}
						B_{1z}\\
						B_{2z}\\
						B_{3z}\\
						B_{4z}
					\end{array}\right]\\
					C_{z}&=diag(C_{z1}, C_{z1}, C_{z1}, C_{z2})
				\end{align*}
				with
				\begin{alignat*}{1}
					A_{1z} & =\left[\begin{array}{cccc}
						0_{3\times 1} & \mathbf{I}_3\\
						0 & 0_{1\times 3}
					\end{array}\right],\ A_{2z}=\left[\begin{array}{cc}
						0 & 1\\
						0 & 0
					\end{array}\right],\\
					B_{1z} & =\left[\begin{array}{cccc}
						0_{3\times 1} & 0_{3\times 3}\\
						1 & 0_{1\times 3}
					\end{array}\right],\ B_{2z}=\left[\begin{array}{cccc}
						0_{3\times 1} & 0_{3\times 1} & 0_{3\times 2}\\
						0 & 1 & 0_{1\times 2}
					\end{array}\right]\\
					B_{3z} & =\left[\begin{array}{cccc}
						0_{3\times 2} & 0_{3\times 1} & 0_{3\times 1}\\
						0_{1\times 2} & 1 & 0
					\end{array}\right],\ B_{4z}=\left[\begin{array}{cccc}
						0_{1\times 3} & 0\\
						0_{1\times 3} & 1
					\end{array}\right],\\
					C_{1z} & =\left[\begin{array}{cccc}
						1 & 0 & 0 & 0\end{array}\right],\ C_{2z}=\left[\begin{array}{cc}
						1 & 0\end{array}\right]
				\end{alignat*}
			\end{proof}
			To implement MPC, the linear equivalent model in \eqref{eq:28} and \eqref{eq:29},
			is discretized as follows:
			\begin{equation}
				z_{d}(k+1|k)=A_{d}z_{d}(k)+B_{d}v_{d}(k).\label{eq:80}
			\end{equation}
			where $\delta$ is the sampling period, $A_{d}=(I+\delta A_{z})$, $B_{d}=\delta B_{z}$. The associated QCQP problem can be formulated as:
			\begin{align}
				J & =\text{min}_{v_{d}(k)}\sum_{i=0}^{N-1}\left[z_{d}(k+i|k)^{\top}Qz_{d}(k+i|k)\right.\nonumber \\
				& \left.+v_{d}(k+i|k)^{\top}Rv_{d}(k+i|k)\right]+z_{d}(N|k)^{\top}\bar{Q}z_{d}(N|k)
				\label{eq:cost}
			\end{align}
			subject to
			\begin{equation}
				z_{d}(k+1)=A_{d}z_{d}(k)+B_{d}v_{d}(k),\forall k=0,\ldots,N-1.\label{eq:42-2}
			\end{equation}
			\begin{gather}
				\left[\begin{array}{c}
					\underline{z}_{1}\\
					\underline{z}_{5}\\
					\underline{z}_{9}\\
					\underline{z}_{13}
				\end{array}\right]\leq\left[\begin{array}{c}
					z_{1}\\
					z_{5}\\
					z_{9}\\
					z_{13}
				\end{array}\right]\leq\left[\begin{array}{c}
					\overline{z}_{1}\\
					\overline{z}_{5}\\
					\overline{z}_{9}\\
					\overline{z}_{13}
				\end{array}\right],\forall k=0,\ldots,N-1.\label{eq:43-2}\\
				v_{\min}\leq v_{d}\leq v_{\max},\forall k=0,\ldots,N-1.\label{eq:44-2}\\
				\triangle\mathcal{H}(z_{d}(k+1))\geq-\gamma\mathcal{H}(z_{d}(k)),k=0,\ldots,N-1.\label{eq:46}\\
				v_{\min}\leq K(A_{d}+B_{d}K)^{i}z_{d}(k+N|k)\leq v_{\max}, \notag\\
				\forall k=0,\ldots,N_{c}.\label{eq:45-2}\\
				\left[\begin{array}{c}
					\underline{z}_{1}\\
					\underline{z}_{5}\\
					\underline{z}_{9}\\
					\underline{z}_{13}
				\end{array}\right]\leq(A_{d}+B_{d}K)^{i}z_{d}(k+N|k)\leq\left[\begin{array}{c}
					\overline{z}_{1}\\
					\overline{z}_{5}\\
					\overline{z}_{9}\\
					\overline{z}_{13} \notag
				\end{array}\right],\\ \forall k=0,\ldots,N_{c}.\label{eq:47}
			\end{gather}
			
			where $N$ is the prediction horizon, $N_{c}$ is the constraint checking horizon, 
			$Q$ and $R$ are the states
			and inputs weight matrices, respectively, and $\bar{Q}$ is the terminal
			weight matrix. The expression in \eqref{eq:42-2} is the dynamical constraint
			of the linear equivalent model. Note that $\underline{z}_{\times}$
			and $\overline{z}_{\times}$ refer to the minimum and maximum value
			associated with the subscript $\times$, respectively. The expression
			in \eqref{eq:44-2} represents control input constraints of $z$ dynamics.
			In the view of the input constraint mapping in \eqref{eq:25-1}, it
			is clear that there is a nonlinear relation between the control inputs
			of $\bar{x}$ and $z$ dynamics. However, one can easily use $\bar{A}^{-1}(\bar{x})$,
			$\mathcal{L}_{\bar{f}}^{4}h_{1}(\bar{x})$, $\mathcal{L}_{\bar{f}}^{4}h_{2}(\bar{x})$,
			$\mathcal{L}_{\bar{f}}^{4}h_{3}(\bar{x})$, and $\mathcal{L}_{\bar{f}}^{2}h_{4}(\bar{x})$
			to verify the boundedness of $v$ will guarantee boundedness of $u$.
			Let us propose the following control barrier function:
			\begin{equation}
				\mathcal{H}_{k}=(z_1-x_{\text{obs}})^{2}+(z_5-y_{\text{obs}})^{2}+(z_9-z_{\text{obs}})^{2}-r_{\text{obs}}^{2}\label{eq:100}
			\end{equation}
			where $x_{\text{obs}}$, $y_{\text{obs}}$, and $z_{\text{obs}}$ describe $x$, $y$ and
			$z$ coordinates of the obstacle, respectively and $0<\gamma\leq1$.
			The expression in \eqref{eq:46} defines the proposed
			safety constraint based on $\mathcal{H}_{k},$ where it forces the
			forward invariance of a safe set $C_{k}$ \cite{ames2019control}.
			Fig. \ref{fig:scheme} shows the proposed MPC-CBF-DFL scheme, where
			the mapping function $\Phi(\overline{x})$ represents the state mapping
			of \eqref{eq:26-1}. $\Phi(\overline{x})$ converts the feedback signals
			of the states in $\overline{x}$ dynamics to $z$ where the optimization
			problem of MPC-CBF is solved providing the values of $v$ to be given
			to the DFL controller.
			
			\medskip{}
			\begin{figure}[h]
				\begin{centering}
					\includegraphics[scale=0.5]{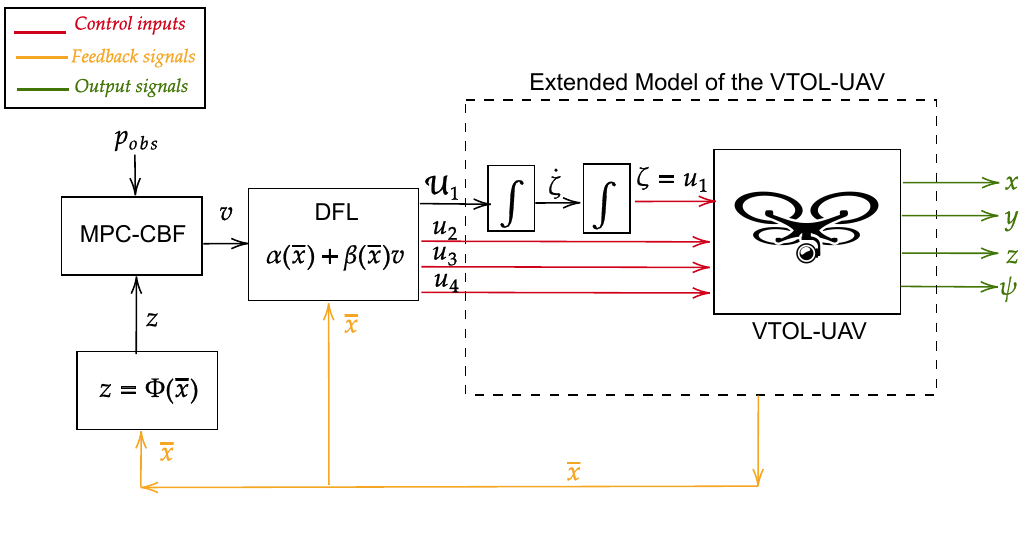}
					\par\end{centering}
				\caption{\label{fig:scheme} MPC-CBF-DFL scheme for the VTOL-UAV.}
				\label{fig:propsedscheme}
			\end{figure}
			
			From \eqref{eq:46}, the level set of CBF constraints is defined by 
		\begin{assum}\label{assum:The-optimal-pblueicted} The optimization problem \eqref{eq:cost}-\eqref{eq:47} is feasible for the initial time k = 0.\end{assum}
		
		\begin{assum}\label{assum:The-pitch-2} The VTOL-UAV pitch angle $\theta$ and roll angle $\phi$
			are constraint by $\frac{-\pi}{2}<\theta<\frac{\pi}{2}$ and $\frac{-\pi}{2}<\phi<\frac{\pi}{2}$ such that
			the inverse of the decoupling matrix in \eqref{eq:25-1} is nonsingular.\end{assum}
		
		\begin{thm}
			\label{thm:thm2} (Asymptotic convergence of MPC-CBF-DFL of VTOL-UAV
			model) The Dynamics of the VTOL-UAV model in \eqref{eq:201} are asymptotically
			stable and the obstacles are avoided using MPC-CBF-DFL scheme if Assumptions
			\ref{assum:The-optimal-pblueicted} and \ref{assum:The-pitch-2} hold true
			and the terminal weight in \eqref{55-1} is selected to be equivalent
			to the infinite horizon cost where
			\begin{equation}
				\sum_{i=0}^{\infty}(\Vert z(i)\Vert_{Q}^{2}+\Vert v(i)\Vert_{R}^{2})=z^{\top}(0)\bar{Q}z(0)\label{55-1}
			\end{equation}
			and
			\begin{equation}
				\bar{Q}-(A_{d}+B_{d}K)^{\top}\bar{Q}(A_{d}+B_{d}K)=Q+K^{\top}RK.\label{eq:56-1}
			\end{equation}
			such that the optimization problem in \eqref{eq:cost}-\eqref{eq:47} is recursively feasible given the K is stabilizing $(A_{d}+B_{d}K)^{i}$ $\forall i=1,\dots N_{c}$  and $\mathcal{H}((A_{d}+B_{d}K)z_{d})>\mathcal(1-\gamma)\mathcal{H}(z_{d})$ $\forall z_d\in\mathcal{Z}_f$ with sufficiently large $N_{c}$ and $0<\gamma\leq1$.
		\end{thm}
		
		\begin{proof}
			Recall the linear equivalent model of the VTOL-UAV in \eqref{eq:28}
			and \eqref{eq:29}, and the cost function in \eqref{eq:cost}. The
			use of the terminal cost function $\bar{Q}$ to solve the Lyapunov
			function in \eqref{eq:56-1} will render the optimal cost function
			for the next time step $k+1$ and thereby $J^{*}(k+1)$ is equal to $J^{*}(k+1)=J^{*}(k)-(\Vert z_{d}(k)\Vert_{Q}^{2}+\Vert v(k)\Vert_{R}^{2})$. 
			Hence, it can be concluded that $J^{*}(k+1)\rightarrow0$ as $k\rightarrow0$. From \eqref{eq:26-1},
			one finds as $z\rightarrow0$ and $\bar{x}\rightarrow0$ results in asymptotic convergence. We prove the recursive stability using the classical terminal constraints \cite{scokaert1996infinite}. 
			Let $\tilde{v}_{d}(k+1)$ denote the input sequence at time $k+1$ corresponding to the optimal prediction
			at time $k$. For feasible $\tilde{v}_{d}(k+1)$ the $N$th element (the tail of $v_{d}^{*}(k+N|k)=Kz_{d}^{*}(k+N|k)$ is required to satisfy the terminal constraint. This is equivalent to the constraints on the terminal state prediction $z_{d}(k+N|k)\in\mathcal{Z}_{f}$, where $\mathcal{Z}_{f}$ is the terminal set. The necessary and sufficient conditions for the predictions generated by the tail $\tilde{v}_{d}(k+1)$ are feasible at time $k+1$ is to control $\mathcal{Z}_{f}$ and achieve safe invariance. The terminal set $\mathcal{Z}_{f}$ is control invariant  if
			$(A_{d}+B_{d}K)z_{d}(k+N|k)\in\text{\ensuremath{\mathcal{Z}_{f}}\hspace{0.1cm}\ensuremath{\forall z_{d}(k+N|k)\in\mathcal{Z}_{f}}}.  
			$ The terminal set $\mathcal{Z}_{f}$ is safe invariant if 
			$
			\triangle\mathcal{H}(z_{d}(k+1|N))\geq-\gamma\mathcal{H}(z_{d}(k|N))\text{\ensuremath{}\hspace{0.1cm}\ensuremath{\forall z_{d}(k+N|k)\in\mathcal{Z}_{f}.}}
			$
			To render $\mathcal{Z}_{f}$ control invariant, we
			need to verify that
			\begin{flalign}
				v_{\text{max}} & \leq K(A_{d}+B_{d}K)^{i}z_{d}(k+N|k)\leq v_{\text{max}}\label{eq:62}\\
				\left[\begin{array}{c}
					\underline{z}_{1}\\
					\underline{z}_{5}\\
					\underline{z}_{9}\\
					\underline{z}_{13}
				\end{array}\right]\leq & (A_{d}+B_{d}K)^{i}z_{d}(k+N|k)\leq\left[\begin{array}{c}
					\overline{z}_{1}\\
					\overline{z}_{5}\\
					\overline{z}_{9}\\
					\overline{z}_{13}
				\end{array}\right]
				& \forall i\geq0\label{eq:64}
			\end{flalign}
			one can design $K$ to render $(A_{d}+B_{d}K)^{i}$ stable such that the
			norm of $|\lambda(A_{d}+B_{d}k)|<1.$ To render $\mathcal{Z}_{f}$ safe invariant, the constraint \eqref{eq:46} will impose that $\mathcal{H}(z_{d}(k+1|N))-\mathcal{H}(z_{d}(k||N))\geq-\gamma\mathcal{H}(z_{d}(k|N))$
			and as a result 
			with $0<\gamma\leq1$. We have $\mathcal{H}(z_{d}(k+1|N))>\mathcal{H}(z_{d}(k|N)$. From \eqref{eq:46}, if $z_{d}(N|0)\in\mathcal{Z}_{f}$
			such that $\mathcal{H}(z_{d}(N|0)\geq0$. One can design $K$ such that   $\mathcal{H}((A_{d}+B_{d}K)z_{d})>\mathcal(1-\gamma)\mathcal{H}(z_{d})$ with $0<\gamma\leq1$ proving the safety invariance of $\mathcal{Z}_{f}$.
		\end{proof}
		Assumption  \ref{assum:The-optimal-pblueicted}
		provides the initial feasibility  conditions at $k=0$ where  Theorem \ref{thm:thm2}  guarantees the feasibility for the next time steps.
		We can define the terminal constraint set as
		$
		{Z}_{f}(N_{c})=\{z_{d}:v_{\text{min}}\leq K(A_{d}+B_{d}K)^{i}z_{d}\leq v_{\text{max}},i=0,1,\ldots N_{c}\}\
		$.
		By choosing sufficiently large $N_{c}$, the allowable operating region of the MPC law will be increased satisfying all the constraints.
		Note that the constraints in \eqref{eq:45-2} and \eqref{eq:47} are extra computational burden specially with large $N_{c}$ which is needed to maintain the feasibility (visit Theorem \ref{thm:thm2}). However, the constraints in \eqref{eq:45-2} and \eqref{eq:47} are linear constraints, where the QCQP presented in \eqref{eq:cost}-\eqref{eq:47} can be solved efficiently by off-the-shelf solvers.
		In the next section, the proposed scheme will be tested with respect
		to the baseline of having constraint of the Euclidean distance (denoted
		as MPC-EC-DFL). Since the chosen CBF function in \eqref{eq:100}
		is equal to Euclidean distance, the new safety constraint in \eqref{eq:46}
		will be $\mathcal{H}(x_{k})>0$.
		\section{Numerical Results\label{sec:Numerical-Results}}
				\begin{figure*}[!htb]
			\begin{centering}
				\centering
				\includegraphics[scale=0.14]{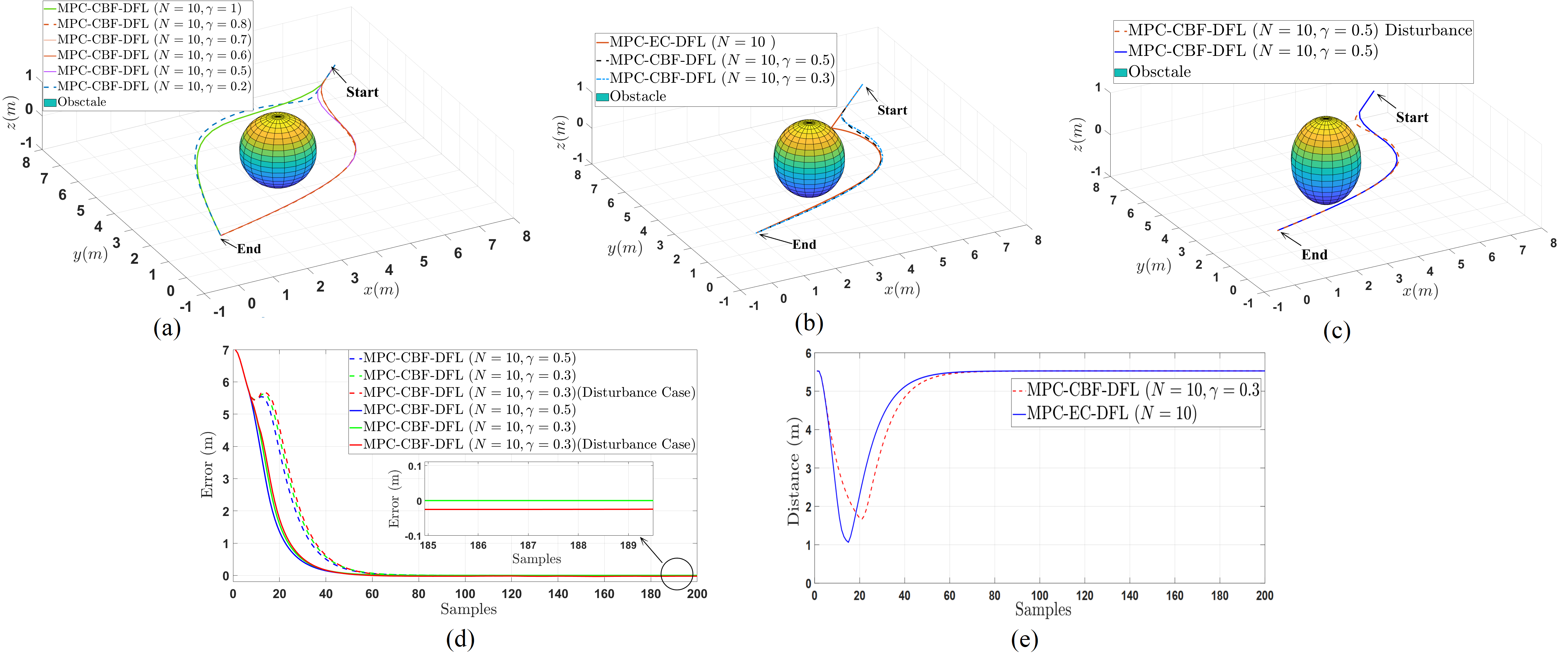}
				
				\par\end{centering}
			\caption{\label{fig:result1-1} Output performance of MPC-CBF-DFL scheme: (a) illustrates the safe navigation of the VTOL-UAV against a spherical Obstacle  with different pairs of $N$ and $\gamma$; (b) presents a comparison between MPC-CBF-DF vs. MPC-ED-DF schemes; (c) shows the trajectories of the proposed scheme against a noisy feedback signal corrupted by a Gaussian noise; (d) shows the error signals where dashed-line and solid-line refer to the $x$ and $y$ errors, respectively; (e) depicts the distance between the VTOL-UAV and the obstacle (MPC-CBF-DF vs. MPC-ED-DF schemes)}.
			
		\end{figure*}
		
		The QCQP problem in \eqref{eq:cost}-\eqref{eq:47}
		was solved using Interior Point OPTimizer (IPOPT) in MATLAB.  The proposed scheme was tested with a sampling time of $\delta=0.05$
		sec where the continuous
		dynamics in \eqref{eq:201} and the control input in \eqref{eq:25-1}
		were numerically integrated using 4th Order Runge Kutta Method and $\delta=0.05$ sec. The VTOL-UAV parameters are defined as follows: $m=0.7\,\text{kg}$, $d=0.3\,\text{m}$, $I_{x}=I_{y}=I_{z}=1.241\,\text{kg}/\text{m}^{2}$.
		The VTOL-UAV started from the initial position $(x=7$, $y=7$,
		and $z=0$) meters and $\psi=0$ rad. The UAV moved toward the origin
		against a sphere obstacle placed at the middle of the map.\\
		
		Fig. \ref{fig:result1-1}.(a) shows the effect of $\gamma$ in
		the evolution of the VTOL-UAV trajectory. 
		By setting a small value of $\gamma$, $\mathcal{H}(x_{k+1})$
		will be significantly greater than $\mathcal{H}(x_{k})$, which simply
		means a bigger safety metric for the next time step $k+1$. Nonetheless,
		Fig. \ref{fig:result1-1}.(a) illustrates the merit of using CBF for obstacle
		avoidance, where stronger safety measures with a lower value of $\gamma$
		would lead to lower predictionhorizons demonstrating the effectiveness
		of employing CBF.
	The performance of the proposed scheme MPC-CBF-DFL against the baseline of Euclidean distance constraint (MPC-ED-DFL) is presented in
	Fig. \ref{fig:result1-1}.(b). MPC-ED-DFL scheme (demonstrated in blue line)
	in Fig. \ref{fig:result1-1}.(b) deviates from the obstacle once it gets
	closer to it. In contrast, using the same prediction horizon the proposed
	MPC-CBF-DFL (plotted in blue and green lines) deviates from the obstacle
	ahead, since the
	proposed scheme relies on the availability of the feedback signals.
Fig. \ref{fig:result1-1}.(c) illustrates the robustness of the proposed scheme
against noisy feedback signals corrupted by a Gaussian noise with
zero mean and variance of 0.05. Fig. \ref{fig:result1-1}.(d) shows the asymptotic convergence of the error of $x$ and $y$ to the origin, while in the case of disturbance the error convergence to the neighborhood of the origin.  Fig. \ref{fig:result1-1}.(e) depicts the distance between the VTOL-UAV and the obstacle using the proposed MPC-CBF-DF scheme and compared to the MPC-ED-DF scheme.

%
\section{Conclusion\label{sec:Sec7_Conclusion}}

In this paper, a cascaded scheme of Dynamic Feedback
Linearization (DFL) and Model prediction Control
(MPC) was proposed to achieve safe navigation of
Vertical Take-off and Landing (VTOL) Unmanned Aerial Vehicles (UAVs).
The MPC was formulated on the linear equivalent model rendered by
the DFL. The proposed scheme showed a strong theoretical guarantee
for closed-loop stability and recursive feasibility based on Linear MPC stability analysis.
Numerical simulations illustrated successful obstacle-free navigation paths
where fast and easy implementation can be achieved with a theoretical
guarantee of convergence. The utilization of the Control Barrier
Functions (CBF) within the safety constraints improved
the obstacle avoidance maneuvers with respect to the constraints based
on Euclidean distances. Numerical results illustrated robustness of the proposed MPC-ED-DFL scheme.

\balance
\bibliographystyle{IEEEtran}
\bibliography{Bib_MPC}
		
	\end{document}